\documentclass[11pt]{amsart}
\usepackage{amssymb,amsmath,amsthm}
\usepackage{a4wide}
\usepackage{graphicx}

\usepackage{color}
\usepackage{mathrsfs}
\usepackage{latexsym}
\usepackage{bbm}
\usepackage{dsfont}

\numberwithin{equation}{section}

\theoremstyle{plain}
\begingroup
\newtheorem{theorem}{Theorem}[section]
\newtheorem{lemma}[theorem]{Lemma}

\endgroup

\theoremstyle{definition}
\begingroup
\newtheorem{definition}[theorem]{Definition}

\endgroup

\usepackage{graphicx}

\usepackage{import}
\usepackage{xifthen}
\usepackage{pdfpages}
\usepackage{transparent}

\newcommand{\VV}{\mathscr{V}}

\newcommand{\Fa}{\mathsf{F}}

\newcommand{\gr}{\mathsf{gr}}

\newcommand{\F}{\mathcal{F}}

\newcommand{\A}{\mathcal A}

\newcommand{\N}{\mathbb{N}}

\newcommand{\R}{\mathbb{R}}
\newcommand{\clos}{\mathrm{clos}}

\newcommand{\E}{\mathcal{E}}

\renewcommand\tilde{\widetilde}

\newcommand{\newatop}{\genfrac{}{}{0pt}{1}}

\def\A{\mathcal{A}}

\newcommand{\G}{\mathsf{G}}

\def\E{\mathcal{E}}

\def\F{\mathcal{F}}

\def\1{\mathbf{1}}

\def\Int{\mathrm{int}}

\def\XXint#1#2#3{{\setbox0=\hbox{$#1{#2#3}{\int}$ }
\vcenter{\hbox{$#2#3$ }}\kern-.57\wd0}}

\def\E{\mathcal{E}}

\newcommand{\Ed}{\mathsf{Ed}}
\newcommand{\Wi}{\mathrm{wire}}
\newcommand{\ext}{\mathrm{ext}}

\newcommand{\inte}{\mathrm{int}}
\newcommand{\X}{\mathsf{X}}

\newcommand{\Per}{\mathrm{Per}}

\usepackage{hyperref}

\title[2D crystallization for soft discs] {A crystallization result in two dimensions for a soft disc affine potential}

\author[G. Del Nin]
{G. Del Nin} 
\address[Giacomo Del Nin]{MPI MiS, Inselstrasse, 22-26 D-04103 Leipzig, Germany
}
\email[G. Del Nin]{giacomo.delnin@mis.mpg.de}

\author[L. De Luca]
{L. De Luca}
\address[Lucia De Luca]{
IAC-CNR, Via dei Taurini, 19 I-00185 Rome, Italy
}
\email[L. De Luca]{lucia.deluca@cnr.it}

\begin{document}

\begin{abstract}
	We prove finite crystallization for particles in the plane interacting through a soft disc potential, as originally shown by C. Radin \cite{Radin_soft}. We give an alternative proof that relies on the geometric decomposition of the energy proved in \cite{DLF1}, and that is based on showing that any minimizer has at least as many boundary points as the canonical ``spiral'' configuration.
\end{abstract}
\maketitle
\tableofcontents
\section{Introduction}
This paper deals with finite crystallization in two dimensions for a soft disc affine pairwise interaction potential at zero temperature.
Following a nowadays standard approach, we want to look at crystallization as a phenomenon emerging by the minimization of suitable energy functionals (see \cite{BL} for a recent review on the crystallization conjecture).
Specifically, we consider an interaction potential of the form
\begin{equation}\label{soft_potential}
\VV^\delta(r):=\left\{\begin{array}{ll}
+\infty&\textrm{if }r<1\,,\\
-1+\frac{r-1}{\delta}&\textrm{if }1\le r\le 1+\delta\,,\\
0&\textrm{if }r>1+\delta\,,
\end{array}
\right.
\end{equation}
with
\begin{equation}\label{delta}
    0< \delta<\frac{1}{2\sin\frac\pi 7}-1\,.
\end{equation}
Given a finite set $\X\subset\R^2$\,,
representing the positions of a system of particles, we define the associated energy as 
\begin{equation}\label{introenergy}
    \E_{\VV^\delta}(\X):=\frac12\sum_{x\in\X}\sum_{x'\in\X\setminus\{x\}}\VV^\delta(|x-x'|)\,,
\end{equation}
where, here and below, $|\cdot|$ denotes the Euclidean norm.
We stress that, in view of assumption \eqref{delta}, any point $x$ belonging to a configuration $\X$ with $\E_{\VV^\delta}(\X)<+\infty$ cannot have more than six nearest neighbors, i.e., there are at most $J\le 6$ points $x_1,\ldots,x_J\in\X$ such that $1\le |x-x_j|\le 1+\delta$ (for $j=1,\ldots,J$). As a consequence, if we were allowed to consider infinite configurations,  we would expect that for any minimizer $\X$ of  $\E_{\VV^\delta}$\,, each point $x\in\X$ has exactly $6$ nearest neighbors and that, in view of \eqref{soft_potential}, the nearest-neighbor distance equals $1$\,.
We show that this behaviour persists in some sense also for finite configurations;
more precisely, we prove that for every fixed $N\in\N$ all the minimizers $\X_N$ of  $ \E_{\VV^\delta}(\X)$\,, among the configurations $\X$ having exactly $N$ points, are subsets of the regular triangular lattice with lattice spacing equal to $1$ and that the {\it canonical configurations} provided by Definition \ref{canondef} below are minimizers of the energy.
Such  canonical configurations have the macroscopic shape of a regular hexagon (which is the Wulff shape for this problem \cite{AFS}) with side length $s$ if $N=N_s=1+6+\ldots+6s=3s^2+3s+1$\,, whereas for general $N\in (N_s,N_{s+1})$ they are obtained by 
nestling the remaining points around the
boundary of the regular hexagon constructed for $N_s$\,.

Our analysis extends \cite{Radin_soft} (see also \cite{Wagner}), where the above result was proved for $\delta=\frac{1}{24}$ using an ansatz on the value of the minimal energy (coinciding with the energy of the canonical configurations),  previously found by Harborth in \cite{Har}. Such an approach had been previously exploited in \cite{HeRa} to treat the {\it sticky disc} model, corresponding to the choice $\delta=0$ in \eqref{soft_potential}. We stress that the energy
 $\E_{\VV^\delta}$ (for $0\le \delta<1/{(2\sin\frac\pi 7)}-1$) is somehow the only pairwise interaction energy for which $2d$ finite crystallization in the regular triangular lattice is proven. We refer to \cite{Theil} for a thermodynamic crystallization result (i.e., when $N\to \infty$)
in the context of long-range Lennard-Jones potentials.

Our approach to the soft disc crystallization problem is based on induction on $N$ and represents a
sort of a hybrid between the ansatz adopted in \cite{Radin_soft} and the (more transparent) strategy introduced in \cite{DLF1}. As in the former, we are using the ansatz on the value of the minimal energy, 
since we prove (and hence, assume in the inductive scheme) that canonical configurations are minimizers of the energy.
On the other hand, adopting the perspective of the latter, our proof avoids the numerical computations in \cite[page 371]{Radin_soft}.
We describe in detail our strategy.
First, we  consider the graph structure associated to any finite configuration $\X$ with finite energy, by defining the set $\Ed(\X)$ of bonds associated to $\X$ as the class of pairs $\{x,y\}\subset\X$ with $1\le |x-y|\le 1+\delta$\,.
By the assumption \eqref{delta},  the graph $\G_\X=(\X,\Ed(\X))$ is planar, so that the tools in planar graph theory, recalled in Section \ref{prelimgra},
can be exploited to face our problem. 
We then rewrite the energy $\E_{\VV^\delta}(\X)$
as the sum of a bulk term, $-3\sharp\X$\,, plus a positive term, denoted by $\F_{\VV^\delta}(\X)$ depending on the ``shape'' of the configuration $\X$\,. 
The former encodes
the fact that each point $x\in\X$ can have at most six nearest neighbors.
The latter, which is defined in \eqref{FV}, consists of a perimeter-like term, plus some other terms that take into account both the stored \textit{elastic} energy due to bonds of length different from 1, and the \textit{plastic} energy due to topological defects, namely, non-triangular faces. 
Trivially, the class of minimizers of $\E_{\VV^\delta}$ coincides with that of the minimizers of 
 $\F_{\VV^\delta}$\,.
In Lemma \ref{lb}, we provide a lower bound for the energy  $\F_{\VV^\delta}$ 
of a given  configuration $\X$ in terms of the energy of the configuration $\X\setminus\partial\X$\,, where $\partial\X$ is the set of ``boundary'' points of $\X$\,, i.e., roughly speaking, of the points $x\in\X$ lying on the boundary of the union of the finite faces of the graph $\G_\X$\,. 
Arguing by induction, we prove at once the two parts of the statement, namely that all the minimizers of $\E_{\VV^\delta}$ are subsets of the $1$-spaced regular triangular lattice and that the canonical configurations $\overline{\X}_N$ minimize the energy for any fixed $N\in\N$\,. 
A key to our proof is using Lemma \ref{lb} to
 show that any configuration $\X$ with $\sharp\X=N$ cannot have less boundary points than the canonical configuration $\overline{\X}_N$\,.
Since in $\overline\X_N$ all the interior points, i.e., the points in $\overline\X_N\setminus\partial\overline\X_N$\,, have exactly $6$ nearest neighbors, this must be the case also for any minimizer of the energy, whence we deduce the desired claim.

We highlight that the proof of Theorem \ref{crystHR} works verbatim for the case $\delta=0$\,,
providing a proof of the crystallization for the sticky disc potential that is slightly different from \cite{HeRa} and \cite{DLF1}.
Finally, since the minimizers of the soft affine problem dealt with here coincide with those of the sticky disc problem, the results on the asymptotic Wulff shape \cite{AFS} as well as the estimates on the fluctuations \cite{S,DPS,CL} hold verbatim in our case (see also \cite{DLF2} for a purely discrete result concerning the uniqueness of minimizers).
\vskip20pt
\textsc{Acknoledgments:}  LDL is member of the Gruppo Nazionale per l'Analisi Matematica, la Probabilit\`a e le loro Applicazioni (GNAMPA) of the Istituto Nazionale di Alta Matematica (INdAM).
 \vskip20pt



{\bf Notation:} In what follows $\N$ denotes the set of positive integer numbers and $\N_0:=\N\cup\{0\}$\,.


\section{Preliminaries on planar graphs}\label{prelimgra}

Here we collect some notions and notation on planar graphs that will be adopted in this paper. 

Let $\X$ be a finite subset of $\R^2$ and let $\Ed$ be a given subset of $\mathsf{E}(\X)$\,, where
\begin{equation*}
\mathsf{E}(\X):= \{\{x,y\}\subset \R^2\,:\, x,y\in\X\,,\,x\neq y\}\,.
\end{equation*}

 The pair $\G=(\X,\Ed)$ is called {\it graph};  $\X$ is called  the set of {\it vertices} of $\G$ and $\Ed$ is called the set of {\it edges} (or {\it bonds}) of $\G$\,.

Given $\X'\subset\X$ we denote by $\G_{\X'}$ the {\it subgraph} (or {\it restriction}) of $\G$ generated by $\X'$\,, defined by $\G_{\X'}=(\X',\Ed')$ where $\Ed':=\{\{x',y'\}\in\Ed\,:\, x',y'\in\X'\}$\,.

\begin{definition}\label{conncomp}
We say that two points $x,z\in\X$ are connected and we write $x\sim z$ if there exist $M\in\N$ and a {\it path} $x=y_0,\ldots,y_M=z$ such that $\{y_{m-1},y_m\}\in\Ed$ for every $m=1,\ldots,M-1$\,. We say that  $\G_{\X_1},\ldots,\G_{\X_K}$ with $K\in\N$ are the {\it connected components} of $\G$ if 
$\{\X_1,\ldots,\X_{K}\}$ is a partition of $\X$
and for every $k,k'\in\{1,\ldots,K\}$ with $k\neq k'$ it holds
\begin{align*}
x_k\sim y_k\qquad&\textrm{for every }x_k,y_k\in\X_k\,,\\
x_{k}\not\sim x_{k'}\qquad&\textrm{for every }x_k\in\X_k\,, x_{k'}\in\X_{k'}\,.
\end{align*}
If $\G$ has only one connected component we say that $\G$ is {\it connected}.
\end{definition}

We say that $\G$ is planar if for every pair of (distinct) bonds $\{x_1,x_2\}, \{y_1,y_2\}\in\Ed$, the (open) segments $(x_1,x_2)$ and $(y_1,y_2)$ have empty intersection.

From now on we assume that $\G=(\X,\Ed)$ is planar, so that we can introduce the notion of face (see also \cite{DLF1}).

By a face $f$ of $\G$ we mean any open, bounded, connected component of 
$
\R^2\setminus \big(\X\cup\bigcup_{\{x,y\}\in\Ed}[x,y]\big)$, which is also simply connected; here $[x,y]$ is the closed segment with extreme points $x$ and $y$. We denote by 
$\Fa(\G)$\,, the set of faces of $\G$ and 
 we set
\begin{equation*}
O(\G):=\bigcup_{f\in\Fa(\G)}\clos(f)\,.
\end{equation*}
We define the Euler characteristic of $\G$ as
$$
\chi(\G)=\sharp\X-\sharp\Ed+\sharp\Fa(\G)\,,
$$
and we warn the reader that this may differ from the standard Euler characteristic in graph theory. 
We just remark that if $\chi(\G)=1$\,, then $\G$ is connected.

With a little abuse of language we will say that an edge $\{x,y\}$ lies on a set $E\subset\R^2$ if the segment $[x,y]$ is contained in $E$\,.
We classify the edges in $\Ed$ in the following subclasses:
\begin{itemize}
\item $\Ed^{\Int}$ is the set of {\it interior edges}, i.e., of edges lying on the boundary of two (distinct) faces; 
\item $\Ed^{\Wi,\ext}$ is the set of {\it exterior wire edges}, i.e., of edges that do not lie on the boundary of any face;
\item $\Ed^{\Wi,\mathrm{int}}$ is the set of {\it interior wire edges}, i.e., of edges lying on the boundary of precisely one face but not on the boundary of its closure (or, equivalently, of $O(\G)$)\,;
\item $\Ed^{\partial}$  is the set of {\it boundary edges}, i.e., of edges lying on $\partial O(\G)$\,. 
\end{itemize}
With a little abuse of notation we set $\partial\X:=\{x\in \X\,:\exists y\in\X \text{ such that }\,\{x,y\}\in\Ed^{\partial}\cup \Ed^{\Wi,\ext}\}$\,.
We define the {\it graph-perimeter} of $\G$ as
\begin{equation*}
\Per_{\gr}(\G):=\sharp\Ed^{\partial}+2\sharp\Ed^{\Wi,\ext}\,.
\end{equation*}
According with the definitions introduced above, if $O(\G)$ has simple and closed polygonal boundary and if $\sharp\Ed^{\Wi,\ext}=0$, then $\Per_\gr(\G)=\sharp\partial\X$\,.
We stress that if $\G$ has no edges, then $\Per_{\gr}(\G)=\sharp\partial\X=0$\,.

Analogously, for every face $f\in \Fa(\G)$ one can define the following subclasses of edges delimiting $f$:
\begin{itemize}
\item $\Ed^{\Wi,\inte}(f)$ is the set of edges lying on the boundary of $f$ but not on the boundary of the closure of $f$;
\item $\Ed^{\partial}(f)$ is the set of edges lying on the boundary of the closure of $f$.
\end{itemize}
Therefore, the {\it graph-perimeter} of a face $f$ is defined by
 \begin{equation*}\label{graphgeoper}
\Per_{\gr}(f):=\sharp\Ed^{\partial}(f)+2\sharp \Ed^{\Wi,\inte}(f).
\end{equation*}

Finally, following \cite[Sec. 2.6]{DLF1}, we define the \textit{defect measure} $\mu(\G)$ of the graph $\G$, as the number of additional edges that we need to add to $\G$ to make it triangulated. More precisely: for every face $f$ with $\Per_{\gr}(f)=k$\,, $k\ge 4$\,, we triangulate it by adding $k-3$ edges that connect not already connected vertices and that do not cross each other, thus obtaining a new graph $\overline{\G}$. Then $\mu(\G):=\sharp \Ed(\overline{\G})-\sharp \Ed(\G) $\,.
\section{The soft disc model}\label{soft:sec}
For every $0<\delta<\frac{1}{2\sin\frac\pi 7}-1$ let $\VV^\delta:[0,+\infty)\to [0,+\infty]$ be the function defined in \eqref{soft_potential} and, for every finite $\X\subset\R^2$\,, let $\E_{\VV^\delta}(\X)$ be the corresponding energy functional as defined in \eqref{introenergy}.
For every $N\in\N$ we denote by $\A_N$ the set of $N$-particle configurations with finite energy, i.e., $\mathcal{A}_N:=\{\X\subset\R^2\,:\,\sharp\X=N\,,\,\E_{\VV^{\delta}}(\X)<+\infty\}$ and we set $\A:=\bigcup_{N\in\N}\A_N$\,.

For every $\X\in\A$\,, we denote by  $\G(\X)$ the {\it graph generated by }$\X$\,, i.e., $\G(\X)=(\X,\Ed(\X))$\,, where $\Ed(\X):=\{\{x,y\}\,:\,x,y\in\X\,,\,1\le|x-y|\le1+\delta\}$\,. 
Notice that the finiteness of $\E_{\VV^{\delta}}(\X)$\,, implies that $\G(\X)$ is a planar graph
and that for any given point $x\in\X$ there could be at most six edges lying on $x$\,.
In what follows, with a little abuse of notation, we set $\Per_{\gr}(\X):=\Per_{\gr}(\G(\X))$ and $\chi(\X):=\chi(\G(\X))$\,.
Analogously, we set $ \Fa(\X):=\Fa(\G(\X))$ and we denote by $ \Fa^{\triangle}(\X)$ denotes the set of the triangular faces of $\X$\,, namely the set of faces $f\in\Fa(\X)$ with $\Per_\gr(f)=3$\,.
By \cite[Theorem 3.1]{DLF1}, for any $\X\in\A$ we have that
\begin{equation}\label{geom_deco}
\E_{\VV^\delta}(\X)=-3\sharp\X+\Per_\gr(\X)+\mu(\X)+3\chi(\X)+\E_{\mathrm{el}}(\X)\,,
\end{equation}
where 
\begin{equation*}
\mu(\X):=\sum_{f\notin \Fa^{\triangle}(\X)}(\Per_{\gr}(f)-3)\textrm{ and }\E_{\mathrm{el}}(\X):=\frac{1}{2}\sum_{\newatop{x,y\in\X}{1<|x-y|\le 1+\delta}}(1+\VV^{\delta}(|x-y|))\,.
\end{equation*}
In what follows, for every $\X\in\A$ we set
\begin{equation}\label{FV}
\F_{\VV^\delta}(\X):=\Per_\gr(\X)+\mu(\X)+3\chi(\X)+\E_{\mathrm{el}}(\X)\,,
\end{equation}
so that, in view of \eqref{geom_deco}, minimizing $\E_{\VV^\delta}$ in $\A_N$ is equivalent to minimizing $\F_{\VV^\delta}$ in $\A_N$\,.

Let $\X\in\A$ have simply closed polygonal boundary.
For every $x\in\partial\X$\,, let $\mathsf{I}^{\mathrm{bdry}}(x)$ and  $\mathsf{I}^{\mathrm{inner}}(x)$ be the sets of boundary and interior edges, respectively, emanating from $x$\,. Let moreover $\alpha(x)$ denote the inner angle spanned by the two boundary edges emanating from $x$\,. The following result is the analogous of \cite[Lemma 1]{Radin_soft}.
\begin{figure}
    \centering
    \def\svgwidth{0.35\columnwidth}
    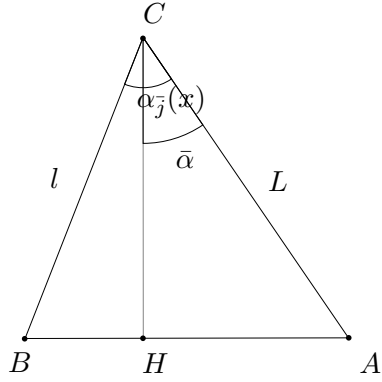

    \caption{Reference figure for Lemma \ref{lemma:minen}.}
    \label{fig:triangle}
\end{figure}
\begin{lemma}\label{lemma:minen}
Let $\X\in\A$ have simply closed polygonal boundary. Then for every $x\in\partial \X$
\begin{equation}\label{true_claim}
\frac{1}{2}\sum_{e\in \mathsf{I}^{\mathrm{bdry}}(x)}\VV^{\delta}(|e|)+\sum_{e\in \mathsf{I}^{\mathrm{inner}}(x)}\VV^{\delta}(|e|)\ge- \frac{\alpha(x)}{\frac\pi 3}\,.
\end{equation}
Moreover, if equality in \eqref{true_claim} holds true then $\alpha(x)=(\sharp\mathsf{I}^{\mathrm{inner}}(x)+\sharp \mathsf{I}^{\mathrm{bdry}}(x)-1)\frac\pi 3$ and $|e|=1$ for all $e\in  \mathsf{I}^{\mathrm{bdry}}(x)\cup \mathsf{I}^{\mathrm{inner}}(x)$\,.
\end{lemma}
\begin{proof}
Let $x\in\partial \X$ be fixed and let
 $I(x):=\sharp \mathsf{I}^{\mathrm{bdry}}(x)+\sharp \mathsf{I}^{\mathrm{inner}}(x)$\,. Let moreover $\alpha_1(x)\,, \ldots,\alpha_{I(x)-1}$ denote the $I(x)-1$ angles spanned by the bonds in $\mathsf{I}^{\mathrm{bdry}}(x)\cup\mathsf{I}^{\mathrm{inner}}(x)$\,, in such a way that $\sum_{j=1}^{I(x)-1}\alpha_j(x)=\alpha(x)$\,. If $\alpha_j(x)\ge\frac{\pi}{3}$\,, then \eqref{true_claim} is trivially satisfied.
 Assume now that $\alpha_{\bar\jmath}(x)=(1-z)\frac{\pi}{3}$ for some $0\le z\le1-\frac{6}{\pi}\arcsin\frac{1}{2(1+\delta)}$ and
 let $l\le L$ be the lengths of the two bonds spanning $\alpha_{\bar\jmath}(x)$\,. 
 We notice that \eqref{true_claim} is proven if we show that
 \begin{equation}\label{claim_1}
 \frac 1 2\VV^\delta(L)\ge z-\frac{1}{2}\,.
 \end{equation}
 To this purpose, we first prove that
 \begin{equation}\label{boundlen}
 L\ge \frac{1}{2\sin\big((1-z)\frac\pi 6\big)}\,,
 \end{equation}
which, in view of the monotonicity of $\VV^\delta$\,, yields
\begin{equation}\label{boundlen2}
 \frac 1 2\VV^\delta(L)\ge \frac 1 2 \VV^\delta\Big(\frac{1}{2\sin\big((1-z)\frac\pi 6\big)}\Big)=-\frac 1 2-\frac{1}{2\delta}+\frac{1}{4\delta\sin\big((1-z)\frac\pi 6\big)}\,.
\end{equation}
 Indeed, let $\bar\alpha$ be the angle formed by the segment $AC$ and the segment $CH$ as in Figure \ref{fig:triangle}. Clearly $\bar\alpha\ge\frac{\alpha_{\bar\jmath}(x)}{2}$\,. Moreover, 
 \begin{equation}\label{trigo_1}
 L\cos\bar\alpha=l\cos(\alpha_{\bar\jmath}(x)-\bar\alpha)\,,
 \end{equation}
 and, since $\E_{\VV^\delta}(\X)<+\infty$\,, 
 we have that 
 \begin{equation}\label{trigo_2}
 L\sin\bar\alpha+l\sin(\alpha_{\bar\jmath}(x)-\bar\alpha)\ge 1\,.
 \end{equation}
 By \eqref{trigo_1} and \eqref{trigo_2}, we get
 \begin{equation}\label{deff}
 L\ge \Big(\sin\bar\alpha+\cos\bar\alpha\tan(\alpha_{\bar\jmath}(x)-\bar\alpha)
 \Big)^{-1}=:(f(\bar\alpha))^{-1}\,.
 \end{equation}
  Since 
  $f'(\bar\alpha)=-\cos\bar\alpha\tan(\alpha_{\bar\jmath}(x)-\bar\alpha)(\tan\bar\alpha+\tan(\alpha_{\bar\jmath}(x)-\bar\alpha))<0$ and $\bar\alpha\ge \frac{\alpha_{\bar\jmath}(x)}{2}$\,,
   we have that $f$ has a maximum at $\bar\alpha=\frac{\alpha_{\bar\jmath}(x)}{2}$ and that $f(\frac{\alpha_{\bar\jmath}(x)}{2})=2\sin(\frac{\alpha_{\bar\jmath}(x)}{2})$\,, thus giving \eqref{boundlen}.

 Now, with \eqref{boundlen2} in hand, we observe that claim \eqref{claim_1} is proven if we show that
 \begin{equation}\label{claim_2}
 g(z):=-\frac{1}{2\delta}+\frac{1}{4\delta\sin\big((1-z)\frac\pi 6\big)}-z\ge 0\,.
 \end{equation}
Notice that 
\[
g'(z)=\frac{\pi}{24\delta}\frac{\cos\big((1-z)\frac\pi 6\big)}{\sin^2\big((1-z)\frac\pi 6\big)}-1,
\]
$g'(0)=\frac{\sqrt{3}\pi}{12\delta}-1\ge 0$ for $0<\delta<\frac{1}{2\sin\frac{\pi}{7}}-1$\,, and
\[
g''(z)=\frac{\pi^2}{144\delta}\frac{1+\cos^2\big((1-z)\frac\pi 6\big)}{\sin^3\big((1-z)\frac\pi 6\big)}\ge 0\,.
\]
It follows that $g(z)$ is monotonically increasing in the interval $[0,1-\frac{6}{\pi}\arcsin\frac{1}{2(1+\delta)}]$\,, which together with the fact that $g(0)=0$ implies \eqref{claim_2}. This concludes the proof of \eqref{true_claim} and shows that if equality holds true in \eqref{true_claim}, then $\alpha_{j}(x)=\frac{\pi}{3}$ for every $j=1,\ldots,I(x)$\,. But this yields 
$$
-I(x)+1=\frac{1}{2}\sum_{e\in \mathsf{I}^{\mathrm{bdry}}(x)}\VV^{\delta}(|e|)+\sum_{e\in \mathsf{I}^{\mathrm{inner}}(x)}\VV^{\delta}(|e|)\ge -I(x)+1\,,
$$
and hence the inequality above is in fact an equality thus providing the last sentence in the statement.
\end{proof}
The following result, which is a consequence of Lemma \ref{lemma:minen}, is the analogous of \cite[Lemma 4.2]{DLF1} in the soft affine case.
\begin{lemma}\label{lb}
Let $\X\in\A$ be connected and have simple and closed polygonal boundary and suppose that $\X':=\X\setminus\partial\X$ is non-empty. Then,
\begin{equation}\label{form:lb}
\F_{\VV^\delta}(\X)\ge \F_{\VV^\delta}(\X')+6\,.
\end{equation}
Moreover, if equality holds true, then $\alpha(x)=(\sharp\mathsf{I}^{\mathrm{inner}}(x)+\sharp \mathsf{I}^{\mathrm{bdry}}(x)-1)\frac\pi 3$ for every $x\in\partial\X$\,,
 $|e|=1$ for every $e\in\bigcup_{x\in\partial\X}(\mathsf{I}^{\mathrm{bdry}}(x)\cup \mathsf{I}^{\mathrm{inner}}(x))$\,, and $\mu(\X)=\mu(\X')$\,.
\end{lemma}
\begin{proof}
By \eqref{geom_deco}, we get
\begin{equation*}
\begin{aligned}
&\sum_{x\in\partial\X}\Big(\frac{1}{2}\sum_{e\in\mathsf{I}^{\mathrm{bdry}}(x)}\VV^\delta(|e|)+\sum_{e\in\mathsf{I}^{\mathrm{inner}}(x)}\VV^\delta(|e|)\Big)\\
\le&\,\E_{\VV^\delta}(\X)-\E_{\VV^\delta}(\X')
=-3\sharp\partial\X+\F_{\VV^\delta}(\X)-\F_{\VV^\delta}(\X')\,,
\end{aligned}
\end{equation*}
whence, using that $\X$ is connected, we deduce that
\begin{equation}\label{uno}
\F_{\VV^\delta}(\X)-\F_{\VV^\delta}(\X')
\ge 3\sharp\partial\X+\sum_{x\in\partial\X}\Big(\frac{1}{2}\sum_{e\in\mathsf{I}^{\mathrm{bdry}}(x)}\VV^\delta(|e|)+\sum_{e\in\mathsf{I}^{\mathrm{inner}}(x)}\VV^\delta(|e|)\Big)\,.
\end{equation}
By \eqref{uno}, in view of Lemma \ref{lemma:minen}, we obtain
\begin{equation}\label{due}
\F_{\VV^\delta}(\X)-\F_{\VV^\delta}(\X')\ge3\sum_{x\in\partial\X}\Big(1-\frac{\alpha(x)}{\pi}\Big)=6\,,
\end{equation}
where in the last equality we have used the fact that $\X$ has simple and closed boundary and the Gauss-Bonnet Theorem to deduce that $\sum_{x\in\partial\X}(\pi-\alpha(x))=2\pi$\,. Therefore, \eqref{form:lb} is proven.

Moreover, if the equality holds true, then we should have that \eqref{true_claim} is satisfied with equality for every $x\in\partial\X$\,; by Lemma \ref{lemma:minen} this implies that
$\alpha(x)=(\sharp\mathsf{I}^{\mathrm{inner}}(x)+\sharp \mathsf{I}^{\mathrm{bdry}}(x)-1)\frac\pi 3$ for every $x\in\partial\X$ and
$|e|=1$ for every $x\in\partial\X$ and for any $e\in \mathsf{I}^{\mathrm{bdry}}(x)\cup \mathsf{I}^{\mathrm{inner}}(x)$\,. It follows that all the faces lying on the boundary are equilateral  triangles with unitary side-length. In particular, if equality in \eqref{form:lb} holds true, then $\mu(\X)=\mu(\X')$ so that also the last sentence in the statement is proven.
\end{proof}
In what follows, for every $s\in\N$\,, we denote by $H_s$ the regular hexagon with side-length $s$\,, centered at the origin, and with two horizontal sides. If $s=0$\,, then we set $H_0:=\{0\}$\,.
\begin{definition}[{\bf Canonical configuration}]\label{canondef}
Let $N\in\N$\,. 
If $N=3s^2+3s+1+(s+1)k+j$\,, with $s,k,j\in\N\cup\{0\}$\,,
$0\le k\le 5$ and $0\le j\le s$
\,, then the canonical configuration is given by
\begin{equation*}
\begin{aligned}
\overline\X_N:=&\,\big(H_s\cap\mathcal{T}\big)\cup\Big\{
e^{ir\frac\pi 3}(\alpha_1+\alpha_2e^{i\frac\pi 3})\,:\,\alpha_1,r\in \N_0\,, \alpha_2\in\N\,, \alpha_1+\alpha_2=s+1\,, 0\le r\le k-1
\Big\}\\
&\phantom{\,\big(H_s\cap\mathcal{T}\big)} \cup\Big\{
e^{ik\frac\pi 3}(\alpha_1+\alpha_2e^{i\frac\pi 3})\,:\,\alpha_1\in \N_0\,, \alpha_2\in\N\,, \alpha_1+\alpha_2=s+1\,, \alpha_2\le j
\Big\}\,.
\end{aligned}
\end{equation*}
This amounts to considering a big regular hexagon with side length $s$ filled with particles, plus $k$ additional full sides, plus a final partially filled side with $j$ particles (see Figure \ref{fig:canonical}).
\begin{figure}
    \centering
\includegraphics[width=0.5\textwidth]{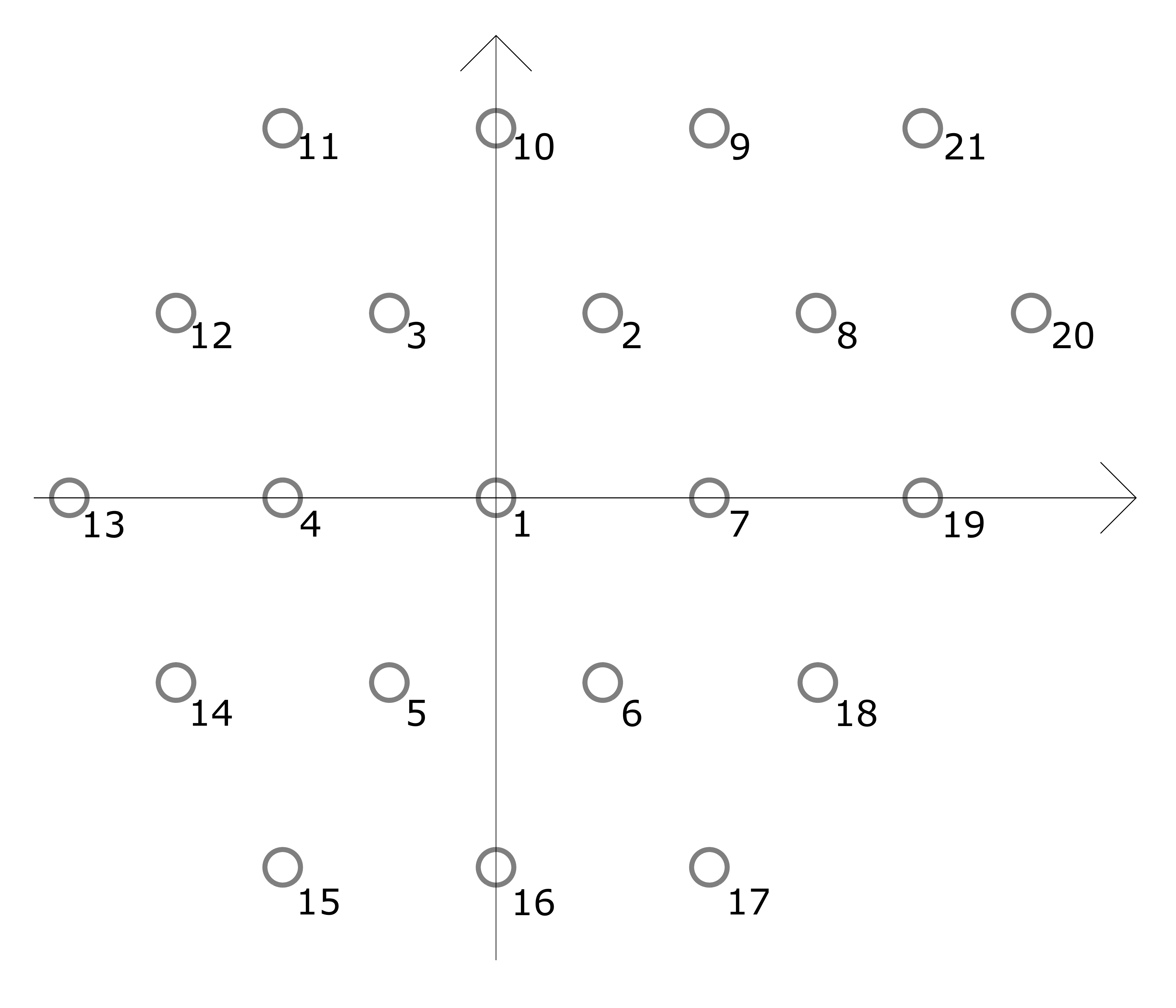}
    \caption{The canonical configurations for $N=1,\ldots,21$\,.}
    \label{fig:canonical}
\end{figure}
By construction, $\mu(\overline\X_N)=\E_{\mathrm{el}}(\overline\X_N)=\sharp\Ed^{\Wi}(\overline{\X}_N)=0$\,, $\chi(\overline{\X}_N)=1$ and 
\begin{equation}\label{percan}
\Per_{\gr}(\overline{\X}_N)=\sharp\partial\overline\X_N=\left\{
\begin{array}{ll}
6s&\textrm{if }N=3s^2+3s+1\\
6s+k+1&\textrm{otherwise}.
\end{array}
\right.
\end{equation}
\end{definition}
\begin{lemma}\label{altribordi}
    For every $N\in\N$, let $\tilde N:=N-\sharp \partial\overline\X_N$. If $N\neq 9$ then the following inequalities hold:
    \begin{enumerate}
        \item $\sharp\partial\overline\X_N\le\sharp\partial\overline\X_{\tilde N}+7$\,.
        \item $\sharp\partial\overline\X_N\le\sharp\partial\overline\X_{\tilde N+1}+6$\,.
    \end{enumerate}
\end{lemma}
\begin{proof}
    First we observe that in the case $N\le 6$ both inequalities are trivially satisfied, so that we can focus on the case $N\ge 7$ for the rest of the proof.

    (i) We divide the proof in a few cases. If $N=3s^2+3s+1$ with $s\ge 1$\,, by \eqref{percan}, we have that $\tilde N=3(s-1)^2+3(s-1)+1$ and hence, again by \eqref{percan}, $\sharp\partial\overline\X_N-\sharp\partial\overline\X_{\tilde N}=6$\,, which proves the claim (i) in this case. 
    
    Let us now consider the case $N=3s^2+3s+1+j$ with $s\ge 1$ and $1\le j\le s$\,. Then, by \eqref{percan}, $\sharp\partial\overline\X_N=6s+1$\,,
    $\tilde N=3(s-1)^2+3(s-1)+1+j-1$\,, 
    and $\sharp\partial\overline\X_{\tilde N}=6(s-1)$ if $j=1$ and
    $\sharp\partial\overline\X_{\tilde N}=6(s-1)+1$ if $j\ge 2$\,.
    Therefore, $\sharp\partial\overline\X_N-\sharp\partial\overline\X_{\tilde N}=6$ for $j\ge 2$\,, and $\sharp\partial\overline\X_N-\sharp\partial\overline\X_{\tilde N}=7$\,, for $j=1$\,.
    This proves the claim (i) 
    also for such a range of parameters.
    
    Now we pass to the case $N=3s^2+3s+1+(s+1)k$ with $s\ge 1$\,, $1\le k\le 5$ and $(s;k)\neq (1;1)$ (the case $s=k=1$ gives $N=9$)\,. Then, by \eqref{percan},
    $\sharp\partial\overline\X_N=6s+k+1$\,, 
    $\tilde N=3(s-1)^2+3(s-1)+1+(s-1+1)(k-1)+s-1$\,, and $\sharp\partial\overline\X_{\tilde N}=6(s-1)+k$\,,
    so that $\sharp\partial\overline\X_N-\sharp\partial\overline\X_{\tilde N}=7$\,, thus proving (i) 
    also in this case.
    
    Finally, we discuss the case $N=3s^2+3s+1+(s+1)k+j$ with $s\ge 1$\,, $1\le k\le 5$\,, and $1\le j\le s$\,. Then, by \eqref{percan}, 
    $\sharp\partial\overline\X_N=6s+k+1$\,,
    $\tilde N=3(s-1)^2+3(s-1)+(s-1+1)k+j-1$\,,
    and $\sharp\partial\overline\X_{\tilde N}=6(s-1)+k+1$\,.
    It follows that $\sharp\partial\overline\X_N-\sharp\partial\overline\X_{\tilde N}=6$\,, thus concluding the proof of (i).

    (ii) Retracing the steps of the proof of (i), we see that the only cases where we need to prove something is when $\sharp\partial\overline\X_N-\sharp\partial\overline\X_{\tilde N}=7$, since in all the other cases this difference is 6, and (ii) follows from the monotonicity inequality $\sharp\partial\overline\X_{\tilde N}\le \sharp\partial\overline\X_{\tilde N+1}$. The cases in which the difference is 7 are: either $N=3s^2+3s+1+1$; or $N=3s^2+3s+1+(s+1)k$ with $s\ge 1$\,, $1\le k\le 5$ and $(s;k)\neq (1;1)$\,. 
    
    In the first case, $\tilde N+1=3(s-1)^2+3(s-1)+1+1$, so that by \eqref{percan} it follows that $\sharp\partial\overline\X_{\tilde N+1}=6(s-1)+1$,  while $\sharp\partial\overline\X_N=6s+1$, so that the claim (ii) follows. 
    
    In the second case, $\tilde N+1=3(s-1)^2+3(s-1)+1+(s-1+1)k$, so that by \eqref{percan} we obtain $\sharp\partial\overline\X_{\tilde N+1}=6(s-1)+k+1$, while $\sharp\partial\overline\X_N=6s+k+1$, and the claim (ii) follows also in this case.
    This concludes the proof of the whole lemma.
\end{proof}
\begin{theorem}\label{crystHR}
Let $N\in\N$ and let $\X_N$ be a minimizer of $\E_{\VV^\delta}$ in $\mathcal{A}_N$. Then $\G(\X_N)$ is connected and, up to rotation and translation, $\X_N$ is a subset of the regular triangular lattice with lattice spacing 1.
Furthermore, $\overline\X_N$ is a minimizer of $\E_{\VV^\delta}$ in $\A_N$\,.
Moreover, if $N\ge 3$\,, then $O(\G(\X_N))$ has simple and closed polygonal boundary, $\Fa(\G(\X_N))=\Fa^{\triangle}(\G(\X_N))$ and $\sharp\Ed^{\mathrm{wire,ext}}(\G(\X_N))=0$\,. 
\end{theorem}
\begin{proof}
We preliminarily notice that the claim is satisfied for $N=1$ and for $N=2$\,. In the latter case the minimizer is given by two points at distance equal to one (i.e., by the canonical configuration $\overline\X_2$). Therefore we focus on the case $N\ge 3$\,. 

First, we observe that $\G(\X_N)$ is connected, since otherwise we could translate one of its connected components until we create a new bond of length 1, thus strictly decreasing the energy. 
Analogously, it is easy to see that $\G(\X_N)$ does not contain wire edges.
Moreover, $\G(\X_N)$ has simply closed polygonal boundary $\Gamma$: if not, we could choose a self-intersection point $p$ of $\Gamma$ and rotate one of the components of $O(\G(\X_N))\setminus\{p\}$ around $p$, until we form another bond of length one, strictly decreasing the energy.

It is immediate to check that for $N=3$ and $N=4$\,, the unique (up to rotations and translations) minimizer is given by the canonical configuration.
Let us discuss the case $N=5,6$\,, by showing first that $\sharp\partial\X_N=N$\,. To this aim, we fix any point $\bar x\in\X_N$, and we consider the half-lines $\ell_1,\ldots,\ell_{N-1}$ starting from $\bar x$ and passing through the other points of $\X_N$ (even those not connected to $\bar x$ by a bond). Let moreover $\alpha_1(\bar x)\le \ldots\le \alpha_{N-1}(\bar x)$ denote the amplitude of the angles formed by  two consecutive half-lines. Then $\alpha_{N-1}(\bar x)\ge \frac{2\pi}{N-1}\ge \frac{2\pi}{5}$\,. If $W$ denotes the corresponding open wedge delimited by the half-lines defining $\alpha_{N-1}(\bar x)$\,, then we have that $W\cap O(\G)=\emptyset$\,, since the maximum angle that can appear in a triangular face is smaller than $\tfrac{2\pi}{5}$\,. This directly proves that $\bar x$ is not an interior point, thus showing that $\sharp\partial\X_N=N$\,. 

Since $\sharp\partial\X_N=N\ge \sharp\partial\overline{\X}_N=\Per_{\gr}(\overline\X_N)=(\Per_{\gr}+\mu+\E_{\mathrm{el}})(\overline\X_N)$\,, we get that $\mu(\X_N)=\E_{\mathrm{el}}(\X_N)=0$\,, i.e., the claim.

Finally, we consider $N\ge 7$\,, and we prove the statement by induction on $N$.
We 
first show that $\sharp\partial\X_N\ge\sharp\partial\overline{\X}_N$\,.
Indeed, assume by contradiction that \begin{equation}\label{falsa}
\sharp\partial\X_N\le\sharp\partial\overline{\X}_N-1\,.
\end{equation}
Since $N\ge 7$\,, we have that $N':=N-\sharp\partial\X_N\ge N-\sharp\partial\overline{\X}_N+1\ge 2$\,.
Moreover, we set $\tilde N:=N-\sharp\partial\overline\X_N\le  N'-1$\,. 
Assume first that $N\neq 9$\,.
Then, by Lemma \ref{altribordi}(i), we have that $\sharp\partial\overline\X_N\le\sharp\partial\overline\X_{\tilde N}+7\le \sharp\partial\overline\X_{N'}+7$\,, so that $\sharp\partial\X_N\le\sharp\partial\overline{\X}_{N'}+6$\,.
Recall that $\X'_{N}=\X_N\setminus\partial\X_N$\,, so that $\sharp\X'_{N}=N'$\,.
By Lemma \ref{lb} and using the inductive assumption that $\overline{\X}_{N'}$ is a minimizer of $\E_{\VV^\delta}$ (and hence of $\F_{\VV^\delta}$) in $\A_{N'}$\,, we thus deduce that
\begin{equation}\label{incrocio}
\F_{\VV^\delta}(\X_N)\ge \F_{\VV^\delta}({\X}'_N)+6\ge \F_{\VV^\delta}(\overline{\X}_{N'})+6\ge \F_{\VV^\delta}(\overline\X_N)\,,
\end{equation}
where in the last inequality we have used that $\sharp\partial\overline{\X}_N\le\sharp\partial\overline{\X}_{N'}+6$ and 
Lemma \ref{altribordi}.
It follows that all the inequalities above are 
 actually equalities, since $\X_N$ is a minimizer. In particular, $\mu(\X_N)=0$ and $|e|=1$ for every $e\in\Ed(\X_N)$\,; but this implies that $\sharp\partial\X_N=\sharp\partial\overline{\X}_{N}$\,, thus contradicting \eqref{falsa}.

Let us consider now the case $N=9$\,. Therefore, if \eqref{falsa} holds true, then $N':=9-\sharp\partial\X_9\ge 2$ and $\tilde N:=9-\sharp\partial\overline\X_9=1$ so that $\sharp\partial\overline\X_9=\sharp\partial\overline\X_{\tilde N}+8\le \sharp\partial\overline\X_{N'}+6$\,, where in the last inequality we have used that $N'\ge 2$ so that, by \eqref{percan}, $\sharp\partial\overline\X_{N'}\ge 2\ge 2+\sharp\partial\overline\X_{\tilde N}$\,.
By \eqref{falsa}, we thus get that $\sharp\partial\X_9\le \sharp\partial\overline\X_{N'}+5$\,.
Therefore, by arguing as in \eqref{incrocio}, we get again a contradiction.
\vskip5pt
If follows that in any case $\sharp\partial\X_N\ge\sharp\partial\overline{\X}_N$\,.
Then,
$$
(\Per_{\gr}+\mu+\E_{\mathrm{el}})(\X_N)\ge\sharp\partial\X_N\ge\sharp\partial\overline{\X}_N=(\Per_{\gr}+\mu+\E_{\mathrm{el}})(\overline\X_N)\,,
$$
which implies that the inequalities above are actually equalities and hence that $\mu(\X_N)=\E_{\mathrm{el}}(\X_N)=0$ and that $\sharp\partial{\X}_N=\sharp\partial\overline{\X}_N$\,.
Therefore, $\overline\X_N$ is a minimizer of $\E_{\VV^\delta}$ in $\A_N$\,, thus concluding the proof of the theorem.
\end{proof}

\end{document}